\documentclass[12pt,reqno,a4paper]{amsart}
\usepackage [latin1]{inputenc}
\usepackage{amscd,amsfonts,amssymb,amsmath,amsthm}
\usepackage[arrow,matrix]{xy}
\usepackage{graphicx,epstopdf}
\usepackage{subfigure,tikz}
\usepackage{color}
\usepackage{mathrsfs,euscript,hyperref}

\newtheorem{thm}{Theorem}[section]
\theoremstyle{remark}
\newtheorem{rk}{Remark}[section]

\theoremstyle{definition}

\newcommand{\xbar}[1]{%
   \kern0.15ex\hbox{%
     \vbox{%
       \hrule height 0.4pt 
       \kern0.3ex
       \hbox{%
         \kern-0.15em
         \ensuremath{#1}%
         \kern-0.1em
       }%
     }%
   \kern0.25ex}%
}

\newcommand{\xbarscr}[1]{%
   \kern0.2ex\hbox{%
     \vbox{%
       \hrule height 0.3pt 
       \kern0.2ex
       \hbox{%
         \kern-0.1em
         \ensuremath{{}_{#1}}%
         \kern-0.15em
       }%
     }%
   \kern0.25ex}%
}

\newcommand{\xbarscrscr}[1]{%
   \kern0.2ex\hbox{%
     \vbox{%
       \hrule height 0.3pt 
       \kern0.2ex
       \hbox{%
         \kern-0.1em
         \ensuremath{{}_{{}_{#1}}}%
         \kern-0.15em
       }%
     }%
   \kern0.25ex}%
}

\newcommand{\checkbxi}{\lefteqn{\boldsymbol{\xi}}\kern.13pc\check{\phantom{\xi}}\kern-.09pc}

\newtheorem{lemma}[subsection]{Lemma}
\newtheorem{cor}[subsection]{Corollary}

\numberwithin{equation}{section}

\usepackage{cancel}
\usepackage[normalem]{ulem}

\begin{document}

\title[Thermodynamics of DNA-RNA renaturation]{
Thermodynamics of DNA-RNA renaturation}

\author{U. A. Rozikov}

\address{ U.Rozikov$^{a,b,c}$\begin{itemize}
 \item[$^a$] V.I.Romanovskiy Institute of Mathematics of Uzbek Academy of Sciences;
\item[$^b$] AKFA University, 1st Deadlock 10, Kukcha Darvoza, 100095, Tashkent, Uzbekistan;
\item[$^c$] Faculty of Mathematics, National University of Uzbekistan.
\end{itemize}}
\email{rozikovu@yandex.ru}

\begin{abstract}  We consider a new model which consists of a DNA together with a RNA. Here we assume that
 DNA is from a mammal  or bird but RNA comes from a virus. To study thermodynamic properties of this
 model we use methods of statistical mechanics, namely, the theory of Gibbs measures.
 We use these measures to describe phases (states) of the DNA-RNA system.
Using a Markov chain (corresponding to Gibbs measure) we give conditions (on temperature) of DNA-RNA renaturation.
\end{abstract}
\maketitle



\section{Introduction}

Each molecule of DNA is a double helix formed by two complementary strands of nucleotides
held together by hydrogen bonds between $G+C$ and $A+T$ base pairs, where  $C$=cytosine, $G$=guanine, $A$=adenine,
and $T$=thymine. Duplication of the genetic
information occurs by the use of one DNA strand as a template for formation of a complementary strand.
The genetic information stored in an organism's DNA contains the instructions for all the proteins
the organism will ever synthesize. It is known that (see, for example, \cite{book}) genetic information
is carried in the linear sequence of nucleotides in DNA. Many experimental and theoretical
works have brought quantitative insights into DNA
base-pairing dynamics that is reviewed in \cite{MD}.

RNA\footnote{https://en.wikipedia.org/wiki/RNA} is a polymeric molecule essential in various biological
roles in coding, decoding, regulation and expression of genes. RNA is assembled as a chain of nucleotides,
but unlike DNA, RNA is found in nature as a single strand folded onto itself,
rather than a paired double strand. Cellular organisms use messenger RNA
 to convey genetic information (using the nitrogenous bases of $C, G, A$ and $U$=uracil) that directs
 synthesis of specific proteins.

 All viruses contain\footnote{https://micro.magnet.fsu.edu/cells/virus.html} nucleic acid, either
 DNA or RNA (but not both), and a protein coat, which encases the nucleic acid. Coronaviruses are a
 group of related RNA viruses that cause diseases in mammals and birds. In humans,
 these viruses cause respiratory tract infections that can range from mild to lethal.

In this paper we study thermodynamic properties of a model which consists a DNA
(from a mammal  or bird)  together with a RNA (from a virus).

Studying DNA's thermodynamics one wants to know how temperature affects
the nucleic acid structure of double-stranded DNA \cite{Man}.
There are few models of thermodynamics
of DNAs (\cite{Ca}, \cite{Pe}, \cite{Ta}). In the recent papers \cite{Rb}, \cite{Rp} we
gave  Ising and Potts models of DNAs and studied their thermodynamics.
Here we shall use the arguments of these papers to study thermodynamic behavior
of a system consisting a DNA and an RNA.

The paper is organized as follows.
In Section 2 we give main definitions and define our model of DNA and RNA.
Moreover, we give a system of functional equations, each solution of which
defines a consistent family of finite-dimensional Gibbs distributions and
guarantees existence of thermodynamic limit for such distributions. These Gibbs
measures are important to describe states of the DNA-RNA system.
Section 3 is devoted to translation invariant Gibbs measures (i.e. constant
solutions of the system of functional equations).
We show uniqueness of translation invariant Gibbs measure (depending on parameters of the model).
In the last section by properties of Markov chains (corresponding to Gibbs measures) we give conditions (on temperature) of DNA-RNA renaturation

\section{System of equations describing of DNA-RNA renaturation}

 The structure of DNA, at the microscopic level, can be described using
 ideas from statistical physics (see \cite{Sw}, \cite{T}),
 where by a single DNA strand is modelled as a stochastic system
 of interacting bases with long-range correlations. This approach makes an
 important connection between the structure of DNA sequence and {\it temperature};
 e.g., phase transitions in such a system may be interpreted as a conformational (topological) restructuring.

 In this section we consider a new model which consists a DNA together with an RNA.
  The bases in nucleic acids can interact via hydrogen bonds.
   Base pairing stabilizes the native three-dimensional structures of DNA and RNA.
 Our interpretation of this system is that
 RNA tries to denature the DNA and renature a new DNA by adding its own nitrogenous bases (as analogue of corona virus's RNA).

 A DNA denaturation process is the breaking of the hydrogen bonds connecting the two stands under treatment by heat \cite{D}, \cite{T}. The
process consists of the splitting of DNA base pairs, or nucleotides, resulting in the separation of two complementary DNA strands\footnote{https://www.ncbi.nlm.nih.gov/books/NBK21514/}.
 In the past decades DNA denaturation  attracted the interest of various researchers, which introduced and studied statistical and dynamical models of this fundamental biological process (see \cite{Ku},   recent paper \cite{DNA} and the references therein).

It is known that in a DNA each $A+T$  pair connected by two hydrogen bonds, while each $C+G$ pair connected
by three hydrogen bonds. Therefore in this section we model them as (spin value) $2=A+T$, $3=C+G$.
A melted (broken) under treatment by heat
hydrogen bond assigned to (spin value) $0$.
The base pairs $A+T$ (in DNA), and $A+U$ (in RNA) considered as identical in process of renaturation of DNA from the RNA (of the virus).

Then a DNA can be considered as a ladder shown in Fig. \ref{fi2}. An RNA is a one-dimensional line (also showed in Fig. \ref{fi2}).
Thus our (spin) system is a double-ladder levels of which denoted by integer numbers $n\in \mathbb Z$.
Assume a base pair is either broken or intact.

\begin{figure}[h]
   \includegraphics[width=13cm]{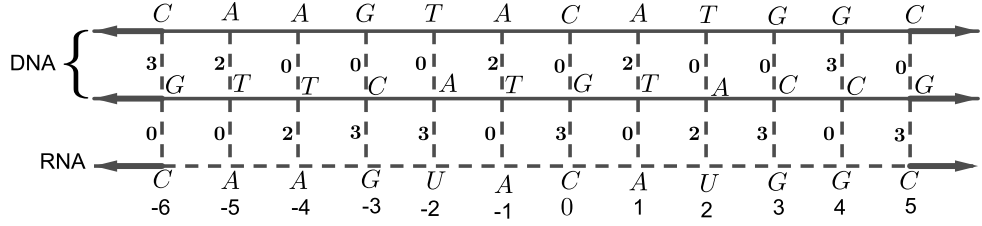}\\
  \caption{The common picture of DNA and RNA (the double-ladder). Configuration consisting $0, 2, 3$ is the state of DNA-RNA denaturation-renaturation at a given temperature $T$. The value $0$ means that the corresponding pair is broken (melted).}\label{fi2}
\end{figure}

To each base pair of level $i\in \mathbb Z$ assign two parameters $d_i$ (to base pair of DNA)  and
$r_i$ (to the base pair of renatured DNA, i.e. between old DNA and RNA). These parameters are defined as
$$d_i=\left\{\begin{array}{lll}
0, \ \ \mbox{if the} \ \ i\mbox{th base pair of DNA is broken}\\[2mm]
2, \ \ \mbox{if the} \ \ i\mbox{th base pair of DNA is intact and at state} \ \ A+T\\[2mm]
3, \ \ \mbox{if the} \ \ i\mbox{th base pair of DNA is intact and at state} \ \ C+G,
\end{array}\right.
$$

$$r_i=\left\{\begin{array}{lll}
0, \ \ \mbox{if the} \ \ i\mbox{th base pair between DNA and RNA is broken}\\[2mm]
2, \ \ \mbox{if the} \ \ i\mbox{th base pair between DNA and RNA is intact and at state} \ \ A+T\\[2mm]
3, \ \ \mbox{if the} \ \ i\mbox{th base pair between DNA and RNA is intact and at state} \ \ C+G.
\end{array}\right.
$$
Since RNA (as corona virus) will break base pair of DNA and puts its own pair, we have condition
\begin{equation}\label{ca}
d_ir_i=0, \ \ \mbox{for all level} \ \ i\in \mathbb Z.
\end{equation}
Thus the configuration space $\Omega$ of our system is build by configurations
$$d=\{d_i\in\{0,2,3\}: i\in \mathbb Z\}, \ \ r=\{r_i\in\{0,2,3\}: i\in \mathbb Z\},$$
as
$$\Omega=\left\{\sigma=(d,r)\in \{0,2,3\}^{\mathbb Z}\times \{0,2,3\}^{\mathbb Z}: d_ir_i=0, \forall i\in \mathbb Z\right\}.$$

For each $\sigma\in \Omega$ define its energy (Hamiltonian) by
\begin{equation}\label{eh}
H(\sigma)=H(d,r)=-J\sum_{i=-\infty}^{+\infty}\left(\delta(d_i, d_{i+1})+
\delta(r_i, r_{i+1})\right)-\alpha\sum_{i=-\infty}^{+\infty}\left(d_i+r_i\right),
\end{equation}
where $J\in \mathbb R$ is coupling constant between base pairs,
$\alpha\in \mathbb R$ is external field and $\delta$ is Kronecker delta:
$$\delta(a,b)=\left\{\begin{array}{ll}
1, \ \ \mbox{if} \ \ a=b\\[2mm]
0, \ \ \mbox{if} \ \ a\ne b.
\end{array}\right.$$

Denote by $\sigma_n$ the restriction of the configuration $\sigma\in \Omega$ on $\mathbb Z_n=\{-n, -n+1, \dots, n-1, n\}$ and by
$\Omega_n$ the set of all such configurations. In general, for a subset $A\subset \mathbb Z$ denote by $\Omega_A$ the set of
all configurations restricted on $A$.

Define a finite-dimensional distribution of a probability measure $\mu$ on $\Omega_n$ as
\begin{equation}\label{d*}
\mu_n(\sigma_n)=Z_n^{-1}\exp\left\{-\beta H_n(\sigma_n)+\sum_{m\in \{-n, n\}}h_{m, d_m, r_m}\right\},
\end{equation}
where $\beta=1/T$, $T>0$ is temperature,  $Z_n^{-1}$ is the normalizing factor,
\begin{equation}\label{hi}
h_{m, i, j}\in \mathbb R, \ \ i, j=0, 2, 3, \ \ m=-n, n
\end{equation} are real numbers and
$$H_n(\sigma_n)=-J\sum_{i=-n}^{n}\left(\delta(d_i, d_{i+1})+
\delta(r_i, r_{i+1})\right)-\alpha\sum_{i=-n}^{n}\left(d_i+r_i\right).$$

We say that the probability distributions (\ref{d*}) are compatible if for all
$n\geq 1$ and $\sigma_{n-1}\in \Omega_{n-1}$:
\begin{equation}\label{**}
\sum_{\omega_n\in \Omega_{\{-n,n\}}}\mu_n(\sigma_{n-1}\vee \omega_n)=\mu_{n-1}(\sigma_{n-1}).
\end{equation}
Here $\sigma_{n-1}\vee \omega_n$ is the concatenation of the configurations.
In this case there exists a unique measure $\mu$ on $\Omega$ such that,
for all $n$ and $\sigma_n\in \Omega_n$,
$$\mu(\{\sigma|_{\mathbb Z_n}=\sigma_n\})=\mu_n(\sigma_n).$$
Such a measure is called a {\it Gibbs measure} corresponding to the Hamiltonian
(\ref{eh}) and values (\ref{hi}).

For simplicity assume that
\begin{equation}\label{sh}
h_{-n, i, j}=h_{n, i, j}, \ \ i, j=0, 2, 3.
\end{equation}
Under this condition the following statement describes conditions on $h_{n, i, j}$ guaranteeing compatibility of $\mu_n(\sigma_n)$.

\begin{thm}\label{di} Probability distributions
$\mu_n(\sigma_n)$, $n=1,2,\ldots$, in
(\ref{d*}) are compatible iff for any $n\geq 1$
the following hold

  \begin{equation}\label{d***}
  \begin{array}{llll}
  x_{n-1}={\theta+\eta^2(\theta^2x_n+y_n)+\eta^3(\theta u_n+v_n)\over
  \theta^2+\theta\eta^2(x_n+y_n)+\theta \eta^3(u_n+ v_n)}\\[3mm]
  y_{n-1}={\theta+\eta^2(x_n+\theta^2y_n)+\eta^3 (u_n+\theta v_n)\over
  \theta^2+\theta\eta^2(x_n+y_n)+\theta \eta^3(u_n+ v_n)}\\[3mm]
 u_{n-1}={\theta+\eta^2(\theta x_n+y_n)+\eta^3(\theta^2 u_n+ v_n)\over
  \theta^2+\theta\eta^2(x_n+y_n)+\theta \eta^3(u_n+ v_n)}\\[3mm]
  v_{n-1}={\theta+\eta^2(x_n+\theta y_n)+\eta^3 (u_n+\theta^2 v_n)\over
  \theta^2+\theta\eta^2(x_n+y_n)+\theta \eta^3(u_n+ v_n)}
  \end{array}
 \end{equation}

Here,
\begin{equation}\label{ap}
\begin{array}{lll}
\theta=\exp(J\beta), \ \ \eta=\exp(\alpha\beta),\\[2mm]
x_n=\exp\left(h_{n,0,2}-h_{n,0,0}\right), \ \  y_n=\exp\left(h_{n,2,0}-h_{n,0,0}\right),\\[2mm]
u_n=\exp\left(h_{n,0,3}-h_{n,0,0}\right), \ \  v_n=\exp\left(h_{n,3,0}-h_{n,0,0}\right).
\end{array}
\end{equation}
\end{thm}
\begin{proof} The proof is similar to the proof of
Theorem 2.1 of \cite{R}.
\end{proof}

It is difficult to find general solutions to (\ref{d***}).

\begin{rk} For $\theta=1$ (i.e. $J=0$) the system (\ref{d***}) has unique solution $x_n=y_n=v_n=u_n=1$.
Therefore below we consider the case $\theta\ne 1$.
\end{rk}

\section{Translation-invariant solutions}
We assume that the unknowns do not depend on $n$, i.e. the value of each unknown is translation invariant. Therefore denote
$$x=\theta \eta^2 x_n, \ \ y=\theta \eta^2 y_n, \ \ u=\theta \eta^3 u_n, \ \ v=\theta \eta^3 v_n.$$
Define mapping
$$F: (x, y, u, v)\in \mathbb R^4_+\to (x', y', u', v')\in \mathbb R_+^4$$
as
 \begin{equation}\label{ds}
  \begin{array}{llll}
  x'=\eta^2 \cdot {\theta^2+\theta^2 x+y+\theta u+v\over
  \theta^2+x+y+u+v}\\[3mm]
  y'=\eta^2 \cdot {\theta^2+ x+\theta^2 y+ u+\theta v\over
  \theta^2+x+y+u+v}\\[3mm]
   u'=\eta^3 \cdot {\theta^2+\theta x+y+\theta^2 u+v\over
  \theta^2+x+y+u+v}\\[3mm]
  v'=\eta^3 \cdot {\theta^2+x+\theta y+ u+\theta^2 v\over
  \theta^2+x+y+u+v}
  \end{array}
 \end{equation}
Then the system (\ref{d***}) is reduced to the finding of fixed points of the mapping $F$, i.e.,
to solving of system $(x, y, u, z)=F(x, y, u, v)$.

Denote
$$M=\{ (x, y, u, v)\in \mathbb R^4_+: x=y, u=v\}.$$
It is easy to see that $F(M)\subset M$, i.e. $M$ is invariant with respect to $F$.

\subsection{Solutions in the set $M$}

Restricting $F$ on $M$ the fixed point problem reduced to the following system

 \begin{equation}\label{xys}
  \begin{array}{ll}
  x=\eta^2 \cdot {\theta^2+(\theta^2+1)x+(\theta+1)u\over
  \theta^2+2x+2u}\\[3mm]
  u=\eta^3 \cdot {\theta^2+(\theta+1)x+(\theta^2 +1)u\over
  \theta^2+2x+2u}.
  \end{array}
 \end{equation}

 From the first equation of this system we get
 \begin{equation}\label{xx}
 (\eta^2(1+\theta)-2x)u=
 2 x^2+(\theta^2-\eta^2(1+\theta^2))x-\eta^2\theta^2.
 \end{equation}

 1) {\it  Case:} $\eta^2(1+\theta)-2 x=0$. In this case we get
 $x={\eta^2(1+\theta)\over 2}$ and substituting this in the RHS of (\ref{xx})
 we get $\eta=\sqrt{{\theta\over 1+\theta}}$.
 Consequently, $x=x_0={\theta\over 2}$. For this value of $\eta$ and $x$ one can explicitly
 find unique positive value of $u=u_0(\theta)$.

 Thus for any $\theta>0$ and $\eta=\sqrt{{\theta\over 1+\theta}}$
 there exists unique solution $(x_0,u_0)$ for (\ref{xys}).

2) {\it Case:} $\eta^2(1+\theta)-2 x\ne 0$.
\begin{equation}\label{u}
u=u(x):={2 x^2+(\theta^2-\eta^2(1+\theta^2))x-
 \eta^2\theta^2\over \eta^2(1+\theta)-2 x}.
 \end{equation}
 Substituting this in the second equation of (\ref{xys}) we get
 \begin{equation}\label{ke}
 Ax^3+Bx^2+Cx+D=0,
 \end{equation}
where
 $$ \begin{array}{llllll}
 A=4(1-\eta),\\[3mm]
 B=2 [(\eta^3-\eta^2-\eta+1) \theta^2+(2\eta^3-\eta-1)\theta+\eta^2(3\eta-1)],\\[3mm]
 C=(-{\eta}^{5}+{\eta}^{3}+\eta^2-1)\theta^{3}+(-2{\eta}^{5}+4\eta^{3}-2\eta^2){\theta}^2+
 (-3\eta^{5}+\eta^{3}+\eta^2)\theta-2{\eta}^{5},\\[2mm]
D={\eta}^{2}\theta^2(\theta-{\eta}^{3}(1+\theta)).
\end{array}
$$
All of the roots of the cubic equation can be found\footnote{https://en.wikipedia.org/wiki/Cubic$_-$equation}.
We are interested in positive solutions, $x_i=x_i(\theta,\eta)$ of the cubic equation.
Moreover, the corresponding $u(x_i)$ defined in (\ref{u}) should be positive too.
Thus condition for parameters $(\theta, \eta)\in \mathbb R^2_+$  of
the existence of positive solutions can be explicitly written $ x_i(\theta,\eta)>0, \ \ u(x_i)>0$.
But the explicit solutions of the cubic  equation have some bulky formulas, therefore we do not present the solution here.
Instead we consider some concrete cases:

3) In the above-mentioned  case 1) we solved the first equation of the system (\ref{xys}) with respect to $u$.
Doing similar argument starting from the second equation of (\ref{xys}) and solving it with respect to $x$ one gets
 $u=u_1={\theta\over 2}$, if $\eta=\sqrt[3]{{\theta\over 1+\theta}}$. Corresponding to $u_1$ one can explicitly
 find unique positive value of $x=x_1(\theta)$.
 Thus for any $\theta>0$ and $\eta=\sqrt[3]{{\theta\over 1+\theta}}$
 we can explicitly give unique solution $(x_1, u_1)$ of (\ref{xys}).

4) {\it Case}: $\eta=1$. In this case the cubic equation is reduced to quadratic equation, which has unique positive solution:
$$x_2=x_2(\theta):={1\over 8}(\theta+2+\sqrt{17\theta^2+4\theta+4}\,).$$
Corresponding $u_2=u(x_2(\theta))$ is also positive. Note that $x_2(\theta)$ and $u_2(\theta)$ have value $+\infty$ (resp. 1/2) if $\theta\to \infty$ (resp. $\theta\to 0$).

Thus in the case $\eta=1$ the system (\ref{xys}) has unique solution $(x_2, u_2)$.

5) Several numerical analysis show that  for $\eta\ne 1$ again we have unique solution (see Fig. \ref{fi3} and \ref{fi3a}).
\begin{figure}[h]
   \includegraphics[width=7cm]{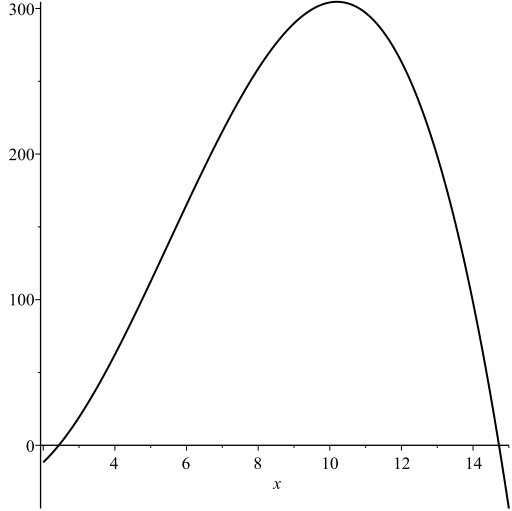}\\
  \caption{The graph of cubic polynomial (\ref{ke}) for $\theta=2$, $\eta=1.2$. In this case there are two positive roots of the polynomial, which approximately:
  $x_1=2.43547$, $x_2=14.72382$ (the negative solution: -0.63929).}\label{fi3}
\end{figure}
\begin{figure}[h]
   \includegraphics[width=7cm]{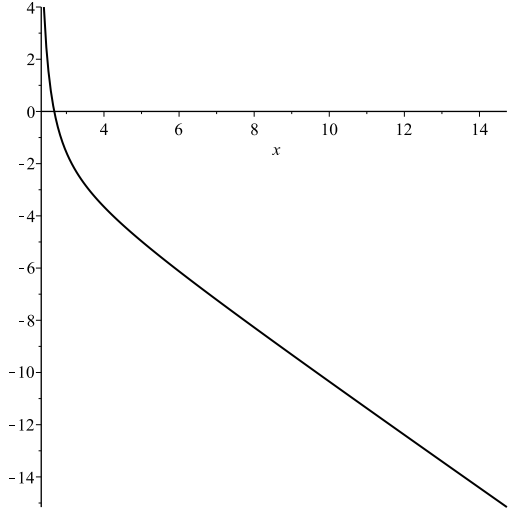}\\
  \caption{The graph of function $u$ defined in (\ref{u}) for $\theta=2$, $\eta=1.2$ on $[2.4, 14.73]$, which contains both positive roots shown in Fig.\ref{fi3}. Thus only one $u_1=u(x_1)$ is positive.}\label{fi3a}
\end{figure}

\subsection{Solutions in the set $\mathbb R_+^2\setminus M$.}
Recall that $F(x, y, u, v)=(x, y, u, v)$ has the form
\begin{equation}\label{da}
  \begin{array}{llll}
  x=\eta^2 \cdot {\theta^2+\theta^2 x+y+\theta u+v\over
  \theta^2+x+y+u+v}\\[3mm]
  y=\eta^2 \cdot {\theta^2+ x+\theta^2 y+u+\theta v\over
  \theta^2+x+y+u+v}\\[3mm]
   u=\eta^3 \cdot {\theta^2+\theta x+y+\theta^2 u+v\over
  \theta^2+x+y+u+v}\\[3mm]
  v=\eta^3 \cdot {\theta^2+ x+\theta y+ u+\theta^2 v\over
  \theta^2+x+y+u+v}
  \end{array}
 \end{equation}
Subtracting from the first equation of this system the second one (resp.  from the third equation of the last one) we get

\begin{equation}\label{db}
  \begin{array}{ll}
  x-y=L \cdot [(\theta-\theta^{-1})(x-y)+(1-\theta^{-1})(u-v)]\\[3mm]
   u-v=\eta L \cdot [(1-\theta^{-1})(x-y)+(\theta-\theta^{-1})(u-v)],
  \end{array}
 \end{equation}
where
$$L\equiv L(x,y,u,v)=\eta^2 \cdot {\theta\over
  \theta^2+x+y+u+v}.$$
  Recall that $\theta\ne 1$.
\begin{lemma}\label{sa} If $(x, y, u, v)$ is a solution to system (\ref{da}) then $x=y$ iff $u=v$.
\end{lemma}
  \begin{proof} Since $L>0$, $\theta\ne 1$, if $x=y$ then from the first equation of (\ref{db}) we get $u=v$. If
   $u=v$ then from the second equation of (\ref{db}) we get $x=y$.
    \end{proof}

 Assume $x\ne y$. Then find $u-v$ from the first equation of (\ref{db})
 and substituting it in the second equation we obtain
   \begin{equation}\label{L}
   (\theta-1)^2(\theta+2)\eta L^2-(1+\eta)(\theta^2-1)L+\theta=0.
   \end{equation}
\begin{thm} The system (\ref{da}) does not have any solution in $\mathbb R_+^2\setminus M$.
\end{thm}
\begin{proof}
  From Lemma \ref{sa} it follows that in $\mathbb R_+^2\setminus M$ may only exist solutions with $x\ne y$ and $u\ne v$.
  Therefore we denote $t={u-v\over x-y}$.

  {\it Case 1:} $t>0$. Assuming $t>0$ from (\ref{db}) we get
  \begin{equation}\label{ne}
  t=\eta\cdot {1+(1+\theta)t\over 1+\theta+t} \ \ \Leftrightarrow \ \ t^2+(1-\eta)(1+\theta)t-\eta=0.
  \end{equation}
  The last equation has unique positive root:
  $$t_1={1\over 2}((\eta-1)(1+\theta)+\sqrt{[(\eta-1)(1+\theta)]^2+4\eta}\,).$$
Thus $u-v=t_1\cdot (x-y)$. Using this from the first equation of (\ref{db}) we get
\begin{equation}\label{ax}1=L\cdot [(\theta-\theta^{-1})+(1-\theta^{-1})t_1]\ \ \Leftrightarrow \ \ L={\eta^2\theta\over
  \theta^2+x+y+u+v}={\theta\over (\theta-1)(1+\theta+t_1)}.
  \end{equation}
One can see that $L$ satisfies (\ref{L}).

By the last formula we get
\begin{equation}\label{dc}
  \begin{array}{llll}
  x=B \cdot [\theta+\theta x+\theta^{-1}y+u+\theta^{-1}v]\\[3mm]
  y=B \cdot [\theta+\theta^{-1} x+\theta y+\theta^{-1} u+v]\\[3mm]
   u=\eta B \cdot [\theta+x+\theta^{-1}y+\theta u+\theta^{-1}v]\\[3mm]
  v=\eta B \cdot [\theta+\theta^{-1} x+y+\theta^{-1} u+\theta v],
  \end{array}
 \end{equation}
where the constant $B$ is
$$B=B(\theta,\eta)= {\theta\over (\theta-1)(1+\theta+t_1)}.$$

 {\it Case 1.1.:} $\theta<1$. It is clear that $(\ref{dc})$ does not have
 any positive solution if $\theta<1$ (because in this case $B<0$).

  {\it Case 1.2.:} $\theta>1$. In this case the system is a linear system of equation of the form $\mathbf{M}\mathbf{v}=\mathbf{b}$, where
  \begin{equation}\label{M}\mathbf{M}=\left(\begin{array}{cccc}
  A&1& \theta&1\\[3mm]
 1&A &1 &\theta\\[3mm]
 \theta &1&C& 1\\[3mm]
  1& \theta & 1& C
  \end{array}\right),
  \ \  A=\theta^2-{\theta\over B}, \, C= \theta^2 -{\theta\over \eta B}, \ \
  \mathbf{v}=\left(\begin{array}{c}
 x\\[3mm]
y\\[3mm]
u\\[3mm]
 v
  \end{array}\right), \ \ \mathbf{b}=-\theta^2\left(\begin{array}{c}
1\\[3mm]
1\\[3mm]
1\\[3mm]
1
  \end{array}\right).\end{equation}

We are interested in positive solutions of the system.
By this system of linear equations we get
\begin{equation}\label{dil}
t_1={u-v\over x-y}={1-A\over \theta-1}={\theta-1\over 1-C} \ \ \Leftrightarrow \ \ AC=A+C+\theta^2-2\theta.
\end{equation}
Using formula of $A$, $C$ and $t_1$ one can see that (\ref{dil}) is satisfied.
Moreover, we have
\begin{equation}\label{di}
\det(\mathbf{M})= (AC+A+C-\theta^2-2\theta)(AC-A-C-\theta^2-2\theta).
\end{equation}
In case $\det(\mathbf{M})\ne 0$ the system  $\mathbf{M}\mathbf{v}=\mathbf{b}$ has unique solution with $x=y$ and $u=v$.
To have its other solutions (with the condition
$x\ne y$ and $u\ne v$) we need to the condition $\det(\mathbf{M})=0$ which by (\ref{dil}) is satisfied and rank$(\mathbf{M})=3$. Solving the linear system $\mathbf{M}\mathbf{v}=\mathbf{b}$, under condition (\ref{dil}), we explicitly obtain infinitely
many solutions:
\begin{equation}\label{mub}
x=- {(\theta^2+v)t_1+v\over t_1(1+t_1)}, \ \ y={v\over t_1}, \ \ u=- {(\theta^2+v)t_1+v\over 1+t_1}, \ \ v>0.
\end{equation}
Thus $x<0$ and $u<0$, i.e., there is no positive solution $x\ne y$ and $u\ne v$.

{\it Case 2:} $t<0$. In this case from (\ref{ne}) we get
   unique negative root:
  $$t=t_2={1\over 2}((\eta-1)(1+\theta)-\sqrt{[(\eta-1)(1+\theta)]^2+4\eta}\,).$$

Thus $u-v=t_2\cdot (x-y)$. Using this from the first equation of (\ref{db}) we get
$$1=L\cdot [(\theta-\theta^{-1})+(1-\theta^{-1})t_2]\ \ \Leftrightarrow \ \ \theta=L\cdot (\theta-1)(1+\theta+t_2).$$

In the last equality,  since $\theta>0$ and $L>0$, it is necessary that $(\theta-1)(1+\theta+t_2)>0$.
It is easy to see that $t_2>-1-\theta$. Consequently, the system may have solution  only for $\theta>1$.
Therefore we have
\begin{equation}\label{ax2}{\eta^2\theta\over
  \theta^2+x+y+u+v}={\theta\over (\theta-1)(1+\theta+t_2)}.
  \end{equation}
By the last formula we get
\begin{equation}\label{dc2}
  \begin{array}{llll}
  x=B_2 \cdot [\theta+\theta x+\theta^{-1}y+u+\theta^{-1}v]\\[3mm]
  y=B_2 \cdot [\theta+\theta^{-1} x+\theta y+\theta^{-1} u+v]\\[3mm]
   u=\eta B_2 \cdot [\theta+x+\theta^{-1}y+\theta u+\theta^{-1}v]\\[3mm]
  v=\eta B_2 \cdot [\theta+\theta^{-1} x+y+\theta^{-1} u+\theta v],
  \end{array}
 \end{equation}
where the constant $B_2$ is
$$B_2=B_2(\theta,\eta)= {\theta\over (\theta-1)(1+\theta+t_2)}.$$

Thus for $\theta>1$ the system is a linear system of equation of the form $\mathbf{N}\mathbf{v}=\mathbf{b}$ where
  \begin{equation}\label{N}\mathbf{N}=\left(\begin{array}{cccc}
  A_2&1& \theta&1\\[3mm]
 1&A_2 &1 &\theta\\[3mm]
 \theta &1&C_2& 1\\[3mm]
  1& \theta & 1& C_2
  \end{array}\right),
  \ \  A_2=\theta^2-{\theta\over B_2}, \, C_2= \theta^2 -{\theta\over \eta B_2}, \ \
  \mathbf{v}=\left(\begin{array}{c}
 x\\[3mm]
y\\[3mm]
u\\[3mm]
 v
  \end{array}\right), \ \ \mathbf{b}=-\theta^2\left(\begin{array}{c}
1\\[3mm]
1\\[3mm]
1\\[3mm]
1
  \end{array}\right).\end{equation}

To have a positive solution of this system it is necessary that
$$A_2=\theta^2-{\theta\over B_2}<0, \, C_2= \theta^2 -{\theta\over \eta B_2}<0.$$
But we have
$$A_2=\theta^2-{\theta\over B_2}<0 \ \ \Leftrightarrow \ \ \theta-{\theta-1\over \theta}(1+\theta+t_2)<0 \ \ \Leftrightarrow \ \ 1<(\theta-1)t_2.$$
The last equality does not hold because $\theta>1$ and $t_2<0$.
Thus  in the case $t<0$ the system (\ref{dc2}) does \emph{not} have any positive solution.
This completes the proof of theorem.
\end{proof}
\section{Gibbs measures: Conditions of DNA-RNA renaturation}

Let $\mu$ be the Gibbs measure corresponding to a solution $(x, y, u, v)$ of the system (\ref{da}).

The measure $\mu$
defines joint distribution
$$\mu\Big(\sigma(n)=(d_n,r_n), \sigma(n+1)=(d_{n+1}, r_{n+1})\Big)= $$
$$
\frac{1}{Z} \exp\Big(J\beta \left(\delta(d_n, d_{n+1})+
\delta(r_n, r_{n+1})\right)+\sum_{m\in \{n, n+1\}}\left[\alpha\beta\left(d_{m}+r_{m}\right)+h_{m, d_m, r_m}\right]\Big),$$
where $Z$ is normalizing factor.

From this,  the relation between the solutions $(x, y, u, v)$ and the transition matrix
for the associated Markov chain  is obtained
from the formula of the conditional probability. The transition matrix of this Markov chain is defined as follows

\begin{equation}\label{pe}\mathbb P=\left( P_{ij}\right)_{i,j=\hat 1, \hat 2,...,\hat 5}=
\left(  \begin{array}{ccccc}
 {\theta^2\over Z_1}& {x\over Z_1}&{y\over Z_1}&{u\over Z_1}&{v\over Z_1}\\[2mm]
 {\theta\over Z_2}&{\theta x\over Z_2}&{\theta^{-1}y\over Z_2}&{u\over Z_2}&{\theta^{-1}v\over Z_2}\\[2mm]
 {\theta\over Z_3}&{\theta^{-1} x\over Z_3}&{\theta y\over Z_3}&{\theta^{-1} u\over Z_3}&{v\over Z_3}\\[2mm]
 {\theta\over Z_4}&{x\over Z_4}&{\theta^{-1}y\over Z_4}&{\theta u\over Z_4}&{\theta^{-1}v\over Z_4}\\[2mm]
 {\theta\over Z_5}&{\theta^{-1} x\over Z_5}&{y\over Z_5}&{\theta^{-1} u\over Z_5}&{\theta v\over Z_5}
  \end{array}\right),
\end{equation}
where
\begin{equation}\label{hat}
\hat 1=(0,0), \ \ \hat 2=(0,2), \ \ \hat 3=(2,0), \ \ \hat 4=(0,3), \ \ \hat 5=(3,0);
\end{equation}
$ Z_1=\theta^2+x+y+u+v, \ \ Z_2= \theta+\theta x+\theta^{-1}y+u+\theta^{-1}v,$
$$Z_3=\theta+\theta^{-1} x+\theta y+\theta^{-1} u+v, \ \
Z_4=\theta+x+\theta^{-1}y+\theta u+\theta^{-1}v, \ \
Z_5=\theta+\theta^{-1} x+y+\theta^{-1} u+\theta v,$$
and $(x,y,u,v)$ a solution of system (\ref{da}) (which depends on both parameters $\theta$ and $\eta$).

Since $\mathbb P$ is a positive stochastic matrix there exists
unique probability vector $\pi=(\pi_1, \dots, \pi_5)$ which satisfies the system of linear equations $\pi \mathbb P=\pi$ (i.e.  $\pi$ is stationary distribution).
Note that this linear system can be explicitly solved. Its solution
$\pi$ depends on both parameters $\theta$ and $\eta$ subject to the constraint that elements must sum to 1.
But coordinates of the vector $\pi$ has a bulky form. Therefore we do not present it here.

The following is a consequence of known (see p. 55 of \cite{Ge}) ergodic theorem for positive
stochastic matrices.
\begin{thm}\label{to} For the matrix $\mathbb P$ defined in (\ref{pe}) and
its stationary distribution $\pi$ the following holds
$$\lim_{n\to \infty} X\mathbb P^n =\pi$$
for all initial probability vector $X$.
\end{thm}

Recall
$r_i=0$ means that RNA does not renature DNA at level $i\in \mathbb Z$.

Thus for a given DNA and RNA we say that the RNA (virus) do not destroy the DNA if
for any $\epsilon>0$ there exists $N\geq 1$ such that for any $n\geq N$ the following inequality holds
$$\mu(\Omega^0_n)>1-\epsilon,$$
where $\mu$ is Gibbs measure corresponding to a solution of system (\ref{d***})
and
$$\Omega^0_n=\{\sigma_n=\{(d_i, r_i)\}\in \Omega_n: r_i=0, \forall i\in \mathbb Z_n\}.$$
Note that each element of $\Omega^0_n$ defines a DNA, which has thermodynamic behavior.

For the measure corresponding to a solution $(x,y,u,v)$  of system (\ref{da}) we have (Markov measure):

$$\mu(\Omega^0_n)=\sum_{\sigma_n\in \Omega^0_n}\pi_{\sigma(-n)}\prod_{i=-n}^n P_{\sigma(i),\sigma(i+1)},$$
where $\pi_{\sigma(-n)}\in \{\pi_1, \dots, \pi_5\}$ (a coordinate of the stationary distribution)
and $P_{\sigma(i),\sigma(i+1)}\in \{P_{\hat i, \hat j}: \ \ {\rm with \ \ odd} \ \ i,j\}$ see (\ref{hat}).

 For a given solution $(x,y,u,v)$ corresponding to it value $\mu(\Omega^0_n)$ depends only on parameters $\theta, \eta$ and $n$, i.e.
 $\mu_n(\theta, \eta):=\mu(\Omega^0_n)$. Using explicit formula of solution $(x,y,u,v)$ one can explicitly
   calculate $\mu_n(\theta, \eta)$. But it will have a bulky form.

For fixed parameters $J$ and $\alpha$ of the model (\ref{eh}) the parameters $\theta$ and $\eta$ are functions of temperature
  $T={1\over \beta}$ (see (\ref{ap})). Therefore for fixed parameters of the model, the quantity $\mu_n(T):=\mu_n(\theta, \eta)$
   is the function of temperature $T$ and $n$ only.   By Theorem \ref{to} and above formulas of matrices we have the following:

\begin{cor} For given parameters $J$ and $\alpha$ of the model (\ref{eh}) RNA of the virus do not destroy DNA (with respect to measure $\mu$)
if \textbf{temperature} $T$ satisfies the following condition:
 $\forall\epsilon>0$ there exists $N\geq 1$ such that  $\forall n\geq N$ the following inequality holds
$$\mu_n(T)>1-\epsilon.$$
\end{cor}
Since $\mu_n(T)$ has a bulky form to check this condition one can use numerical analysis.

\end{document}